\documentclass[letterpaper, 12pt]{article}

\usepackage{url}
\usepackage{stmaryrd}
\usepackage{enumerate}

\usepackage{graphicx}
\usepackage{subfig}
\usepackage[cmex10]{amsmath}
\usepackage{amsthm}
\newtheorem{theorem}{Theorem}
\def\bp{\overline{p}}

\def\ie{{\em i.e.},~}
\def\eg{{\em e.g.},~}
\def\etal{{\em et al.}~}
\def\etc{{\em etc.}~}
\def\pmax{p_{\sf max}}
\def\Dmax{D_{\sf max}}
\def\Umax{U_{\sf max}}
\def\wtp{\widetilde{p}}
\def\wtD{\widetilde{D}}
\def\vbp{\vec{\bp}}

\newcommand{\pd}[2]{\frac{\partial{#1}}{\partial{#2}}}
\newcommand{\pdd}[2]{\frac{\partial^2{#1}}{\partial{#2}^2}}
\newcommand{\pref}[1]{(\ref{#1})}
\renewcommand{\th}[1]{#1^{\rm th}}
\newcommand{\NEP}[1]{{\sf NEP}_{\sf #1}}
\newcommand\pkmax[1]{p_{#1 \sf max}}
\newcommand\Dkmax[1]{D_{#1 \sf max}}
\newcommand\Ukmax[1]{U_{#1 \sf max}}
\renewcommand{\vec}[1]{\mathbf{#1}}

\begin{document}

\title{Application Neutrality and \\ a Paradox of Side Payments}

\author{Eitan Altman \and St\'{e}phane Caron \and George Kesidis}

\maketitle

\begin{abstract}
The ongoing debate over net neutrality covers a broad set of issues related to the regulation of public networks. In two ways, we extend an idealized usage-priced game-theoretic framework based on a common linear demand-response mo\-del \cite{arXiv}. First, we study the impact of ``side payments'' among a plurality of Internet service (access) providers and content providers.  In the non-monopolistic case, our analysis reveals an interesting ``paradox'' of side payments in that overall revenues are reduced for those that receive them. Second, assuming different application types (\eg HTTP web traffic, peer-to-peer file sharing, media streaming, interactive VoIP), we extend this model to accommodate differential pricing among them in order to study the issue of application neutrality. Revenues for neutral and non-neutral pricing are compared for the case of two application types.
\end{abstract}

\section{Introduction}
\label{intro}

Different issues have been raised in the context of the net neutrality debate. 
	For Tim Berners-Lee \cite{tbl}, it means that ``if I pay to connect to the Net with a certain quality of service, and you pay to connect with that or greater quality of service, then we can communicate at that level.''
	For Tim Wu \cite{wu}, the main idea is that ``a maximally useful public information network aspires to treat all content, sites, and platforms equally.''
	According to \cite{hahn}, it ``usually means that broadband service providers charge consumers only once for Internet access, do not favor one content provider over another, and do not charge content providers for sending information over broadband lines to end users.''

These definitions raise different questions, including connectivity, non-dis\-cri\-mi\-na\-tion of application, based on type or origin, and network access pricing. Net neutrality is a subject involving a range of issues regarding the regulation of public networks \cite{chong}: (a) content neutrality, (b) blocking and rerouting, (c) denying IP-network interconnection, (d) network management, and (e) premium service fees. (b) pertains to providers discriminating packets in favor of their own or affiliated content, while (c) is related to agreements between last-mile and backbone providers. (d) has been a central argument for ISPs protesting the enforcement of net neutrality principles: they defend their right to manage their own networks, especially in order to deal with congestion issues (\eg due to high-volume peer-to-peer (P2P) traffic, see the ``Comcast v. the FCC'' decision \cite{comcast}). They claim that regulations would act as a disincentive for capacity expansion of their networks.

In this paper, considering usage-based revenues, we address issues from topic (a):
 side payments among providers and application neutrality. Massive copyright infringements led copyright holders to seek remuneration from ISPs, while congestion due to P2P file sharing led some providers to adopt not application-neutral policies (\eg Comcast throttling BitTorrent traffic) and to consider usage pricing (as a congestion penalty, for overage of a monthly quota, or for premium service, \eg \cite{NOMS10}). In what follows, we study side payments (from Internet Service (access) Providers (ISPs) to Content Providers (CPs), or in the reverse direction) and consider the impact of not application-neutral pricing independent of congestion.

That is, we assume consumers are, to some extent, willing to pay usage-dependent fees. Providers are then competing to settle on their usage-based prices, their goal being to maximize revenues coming from these charges. Note that a null price in the following does not mean a provider has no income, but rather that all their monthly revenues come from flat-rate priced service components. Study of the flat-rate regime is, however, out of the scope of this paper. See, \eg \cite{ciss08} for a comparison of both regimes for a simple model of congestion management.

The rest of the paper is organized as follows. We discuss related work in subsection \ref{subsec:related} and describe our problem framework in section \ref{sec:setup}. In section \ref{sec:side}, we study the impact of side payments on the competition between providers. We extend our framework in section \ref{sec:app} to analyze the effect of not application-neutral pricing by the ISPs. We conclude in section \ref{sec:conclusion}{ and discuss future work}.

\subsection{Related Work}
\label{subsec:related}

Previously, we considered  certain net neutrality related issues like side payments and premium service fees (e), limiting our consideration to monopolistic providers \cite{arXiv}. In the following, we extend this model to include competition between multiple identical providers (actually based on an idea sketched in Section IV of \cite{arXiv}).

{The validity of the ISPs' argument that net neutrality is a disincentive for bandwidth expansion has been studied in \cite{cheng}. In the proposed framework, incentives for broadband providers to expand infrastructure capacity turned out to be higher under net neutrality, with ISPs tending to under- or over-invest in the non-neutral regime.}

Ma \etal \cite{shapley1, shapley2} advocate the use of Shapley values as a fair way to share profits between providers. This approach yields Pareto optimality for all players, and expects in particular CPs, many of whom receive advertising revenues, to take part in network-capacity investments. However, this approach is coalitional and there are many obstacles to its real-life implementation.

{\cite{hande} deals with the question of side payments and deploys a framework in which CPs can subsidize consumers' connectivity costs. The authors compare an unregulated regime with a ``net neutral'' one where restrictions apply on the maximum price ISPs can charge content providers. They find out that, even in the neutral case, CPs can benefit from sharing revenue from end users if the latter are sufficiently price sensitive (and the cost of connectivity is low enough). Their framework is insightful, but does not take CP revenues into consideration.}

In \cite{walrand}, the authors address whether local ISPs should be allowed to charge remote CPs for the ``right'' to reach their end users (again, this is the side payment issue). Through study of a two-sided market, they determine when neutrality regulations are harmful depending on the parameters characterizing advertising rates and consumer price sensitivity\footnote{As in \cite{hande}, the outcome essentially depends on end users' price sensitivity, but here it is furthermore related to CP (advertising) revenues.}.

{

We study a similar issue in section \ref{sec:side}, yet with a significantly different scenario. In \cite{walrand}, the ISPs invest in network infrastructure, and then the CPs invest depending on the quality of the resulting network. The difference in time and scale of these investments justifies a leader-follower dynamics. Now in our model we suppose that the network is already deployed and that providers are setting up usage-based pricing to leverage ongoing revenues. For example, our scenario could be AT\&T beginning to charge Google{\footnote{An intention they voiced, \eg in \cite{att}: ``Now what [the content providers] would like to do is use my pipes free, but I ain't going to let them do that because we have spent this capital and we have to have a return on it.'' 
}} and its customers on a usage-based basis, while the leader-follower scenario would be closer to ISPs investing in optical fiber connections and high-quality video-on-demand providers coming to the new market.

Otherwise, both models share some assumptions, including a fixed number of players, homogeneous providers, and a uniform distribution of consumers among providers once the price war has ended. However, our CPs' revenues come from usage-based fees rather than advertising\footnote{We suggest in \ref{subsec:fw} how to take both into account.}, and our content consumption model is different: users subscribe to one CP and get all their content from him. This could be the case, \eg with online newspapers, music stores, video on demand, \etc Though this setting is more restrictive than a network where users are willing to access all CPs, it is of practical use and fits more to the homogeneity assumption.

}

The net neutrality debate is discussed in \cite{odlyzko} in light of historical precedents, especially dealing with the question of price discrimination. A conclusion about the way customers value the network is that connectivity is far more important than content. 

\section{Problem set-up}
\label{sec:setup}

Our model encompasses three types of players: the Internauts (end users), modeled collectively by their demand response, $n_1$ last-mile ISPs, and $n_2$ CPs. Consumers pay providers usage-dependent fees for service and content that requires one ISP and one CP. Providers then compete in a game to settle on their usage-based prices, which may turn out to be 0\$/byte, \ie only flat-rate fees would apply.

\subsection{Common Demand Response Model}

Let us denote by $p_{1i} \geq 0$ (resp. $p_{2j} \geq 0$) the usage-based price of the $\th{i}$ ISP (resp. $\th{j}$ CP). These prices act as disincentives on consumers' demand for content and bandwidth. We model this with a simple linear response: the amount users are ready to consume, given that they chose ISP $i$ and CP $j$, is
	\begin{equation*}
	D(p_{1i}, p_{2j}) = \Dmax - d_1 p_{1i} - d_2 p_{2j},
	\end{equation*}
where $d_k$ is the demand {\em sensitivity} to price paid to provider of type $k$ (the first subscript $k=1$ for ISP and $k=2$ for CP).
We are dealing here with a set of homogeneous users sharing the same response to price variations. The parameter $\Dmax$ reflects demand under pure flat-rate pricing.

Note that all providers may not measure demand on the same scale: ISPs focus on bandwidth consumption and express demand in bytes, while CPs are concerned with content consumption and/or advertising revenues, thus expressing demand in number of clicks or products sold (books, music albums, {\em etc.}). However, using a single demand metric is very convenient for our purposes here, and other metrics can be approximated from this one using an appropriate scaling factor.

In what follows, we furthermore suppose that users are only concerned with the total usage-based price they are charged, \ie they don't care whether they are giving money to an ISP or a CP. Equal demand sensitivies to price ensue,
\ie $d_1 = d_2 = d$. Since $D\geq 0$, define the maximum price
	\begin{equation*}
	\pmax \ := \ \frac{\Dmax}{d} \ \geq \  p_{1i} + p_{2j}.
	\end{equation*}

\subsection{Customer Stickiness}

As we suppose all providers of a given type propose the \emph{same} type/quality of content/service, user decisions are only based on price considerations. For example, if an ISP charges a price significantly lower than the other ISPs, in the long run all customers will choose it and the others will have no choice but to align their prices or opt out of the game. Therefore, our homogeneity hypothesis means all $n_1$ ISPs (and similarly all $n_2$ CPs) have roughly the same prices:
	\begin{equation*}
	\begin{array}{c c c c c c c}
	p_{11} & \approx & p_{12} & \approx & \cdots & \approx & p_{1 n_1}, \\
	p_{21} & \approx & p_{22} & \approx & \cdots & \approx & p_{2 n_2}.
	\end{array}
	\end{equation*}
As providers play the usage-based pricing game, first-order differences between these prices may appear (\eg \\ the $\th{i}$ ISP reducing his price by $\delta p_{1i}$ to attract new end users). Consumers are then more likely to go to the cheapest providers of each type, but price differences may be too small to convince all of them to move and some will stay with their current provider. This phenomenon is known as \emph{customer stickiness}, \emph{inertia} or \emph{loyalty}. To model it, we define the fraction $\sigma_{ki}$ of users committed to the $\th{i}$ provider of the $\th{k}$ type ($k=1$ for ISPs and 2 for CPs) as a function of $\vec{p}_k = (p_{k1}, \ldots, p_{k n_k})$, \ie $\sigma_{ki} := \sigma(i, \vec{p}_k)$. Properties expected of the ``stickiness function'' $\sigma$ include:
\begin{itemize}
	\item[(a)] 	$\sigma(i, \vec{p}_k) \geq 0$ and $\sum_{j=1}^{n_k} \sigma(j, \vec{p}_k) = 1$;
	\item[(b)] if $\vec{p}_k = (p, p, \ldots, p)$, then $\forall i \in \{1,..., n_k \}, \\ \sigma(i, \vec{p}_k) = 1/n_k$; and
	\item[(c)]  $p_{ki} < p_{kj} \Rightarrow \sigma(i, \vec{p}_k) > \sigma(j, \vec{p}_k)$.
\end{itemize}
In other words, the distribution is uniform if all providers of a given type charge exactly the same price, and ensures otherwise that cheaper providers attract more consumers. 
We chose the following model which satisfies these properties:
	\begin{equation}
	\label{eqn:sigma}
	\sigma(i, \vec{p}_k) \ = \ \frac{1/p_{ki}}{\sum_{j=1}^{n_k} 1/p_{kj}} \ =: \ \sigma_{ki},
	\end{equation}
The average usage-based price charged by a provider of type $k$ for a customer is then $\bp_k := \sum_i \sigma_{ki} p_{ki},$ \ie the harmonic mean of $\{p_{ki}\}_i$.

\subsection{Non-discriminating setting}

In a ``neutral'' setting with no side payments nor application discrimination, the $\th{i}$ ISP's expected usage-based revenue is given by
	\begin{equation*}
	U_{1i} \ = \ \sum_{j=1}^{n_2} \sigma_{1i}\,\sigma_{2j}\,D(p_{1i}, p_{2j})\,p_{1i} \\
		\ = \ \sigma_{1i}\,D(p_{1i}, \bp_2)\,p_{1i},
	\end{equation*}
and similarly $U_{2j}$ for the $\th{j}$ CP. Then,
	\begin{equation*}
	\pd{U_{1i}}{p_{1i}} \ = \ \left[-\sigma_{1i} \sum_{j \neq i} \frac{1}{p_{1j}}
		- \frac{1}{\pmax - p_{1i} - \bp_2}
		+ \frac{1}{p_{1i}}\right] U_{1i}.
	\end{equation*}
A necessary condition for an interior equilibrium is $\pd{U_{1i}}{p_{1i}}(\bp_1, \bp_2) = 0$, which in this case yields $(n_1+1) \bp_1 + \bp_2 = \pmax$. Combining with the similar condition for CPs, we get a non-singular linear system whose solution is
	\begin{equation*}
	\bp_k^* = \frac{n_{3-k}}{n_1 n_2 + n_1 + n_2}\,\pmax \quad \textrm{for } k=1,2.
	\end{equation*}
Demand and revenues are then given by
	\begin{eqnarray*}
	D^* & = & \frac{n_1 n_2}{n_1 n_2 + n_1 + n_2}\,\Dmax, \\
	U_{ki}^* & = & \frac{n_{3-k}^2}{(n_1 n_2 + n_1 + n_2)^2}\,\Umax,
	\end{eqnarray*}
for $k=1,2$ and $i \in \llbracket 1, n_k \rrbracket$. This solution is actually a Nash Equilibrium Point (NEP) since
	\begin{equation*}
	\pdd{U_{ki}}{p_{ki}}(\bp_1^*, \bp_2^*) \ = \ \frac{-2 U_{ki}^*}{\bp_k^* (\pmax - \bp_1^* - \bp_2^*)} \ < \ 0,
	\end{equation*}
which means it is a local maximum in revenue for all players.

As expected, customers benefit from competition among the providers. {See Figure \ref{fig:neutral-D}. }With 2 ISPs and 2 CPs, demand is only $50\%$ of its potential $\Dmax$, while it is about $70\%$ of $\Dmax$ with 5 ISPs and 5 CPs.
	\begin{figure}[ht]
	\begin{center}
	\subfloat[]{%
		\includegraphics[width=0.4\textwidth]{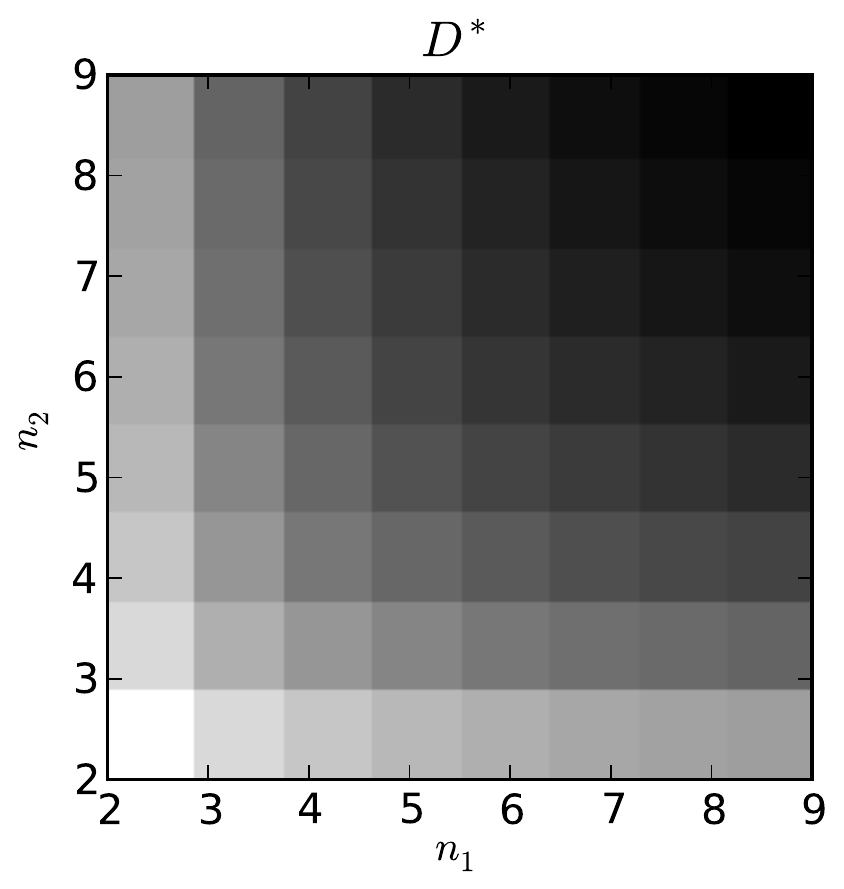}
		\label{fig:neutral-D}}
	\subfloat[]{%
		\includegraphics[width=0.4\textwidth]{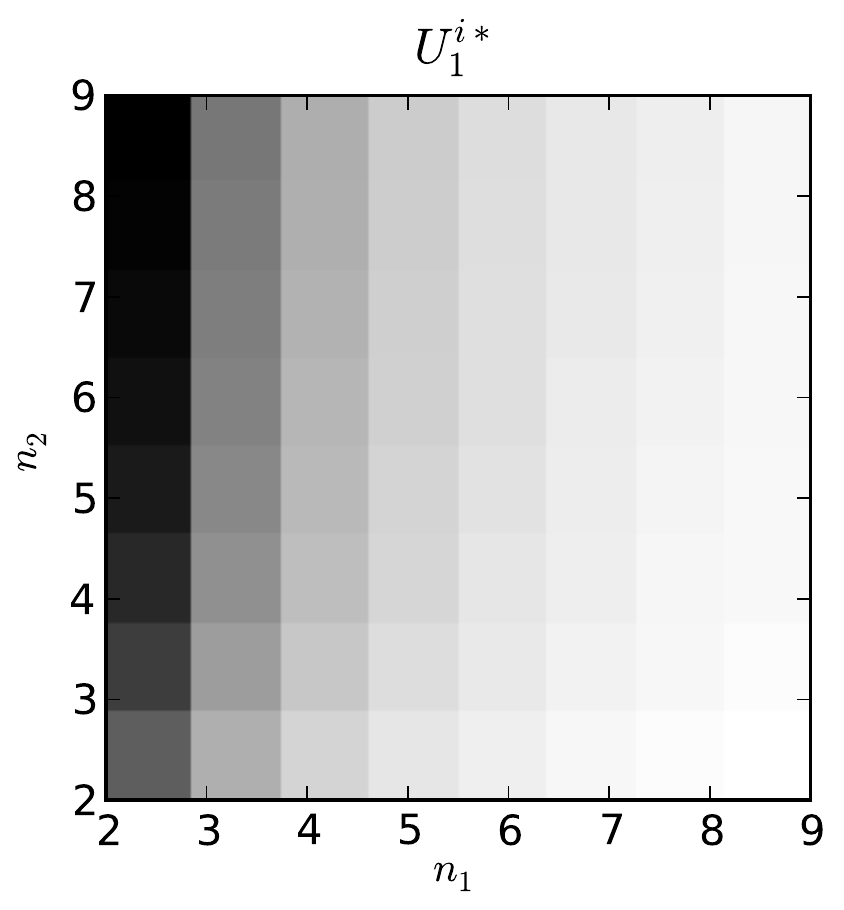}
		\label{fig:neutral-U}}
	\caption{Equilibrium demand and ISP revenues as functions of $n_1$ and $n_2$
	(CP revenues are symmetric). In (a), black means $85\%$ of $\Dmax$ while
	white is $50\%$ of $\Dmax$. In (b), black is $10\%$ of $\Umax$ and white
	$0.5\%$ of $\Umax$.}
	\end{center}
	\end{figure}
This base model also encompasses two expected behaviors: providers of one type benefit from increased competition among those of the other, while their revenues are significantly reduced by increased competition in their own group{ (see Figure \ref{fig:neutral-U})}. Note that competition in a provider's own group has much greater impact on their income than competition in the other.

\section{Side Payments}
\label{sec:side}

Suppose now that there are side payments between the two types of providers. We introduce a usage-based fee $p_s$ from the CPs to the ISPs. When $p_s > 0$, CPs remunerate the ISPs, \eg to support the bandwidth costs or share in advertising 
revenue\footnote{Advertising in delivered content is arguably not specifically requested by end-users.}. On the other hand, if $p_s < 0$, ISPs give money to the CPs, \eg for copyright remuneration. We suppose ISPs or CPs receive side payments collectively and ultimately share the aggregate amount proportionally to their customer shares. Hence, provider revenues become:
	\begin{eqnarray*}
	U_{1i} & = & \sigma_{1i} D(p_{1i}, \bp_2) (p_{1i} + p_s),
		\quad i \in \{ 1, ...,  n_1 \}, \\
	U_{2j} & = & \sigma_{2j} D(\bp_1, p_{2j}) (p_{2j} - p_s),
		\quad j \in \{ 1, ..., n_2 \},
	\end{eqnarray*}
where all demand and price factors are non-negative. 

It is expected that $p_s$ is \emph{not} a decision variable for any player or group of players. Indeed, since revenues are monotonic in $p_s$, those controlling it would always be incented to increase or decrease it (if they are ISPs or CPs respectively), leading the other players to opt out of the competition. Therefore, $p_s$ would normally be regulated and we will consider it a \emph{fixed} parameter from now on.

Necessary conditions for an interior equilibrium (null first-derivatives of revenues) yield:
	\begin{eqnarray}
	\label{side-equ1}
	\left[\frac{\bp_1}{\pmax - \bp_1 - \bp_2} - \frac{1}{n_1}\right] (\bp_1 + p_s) + p_s & = & 0, \\
	\label{side-equ2}
	\left[\frac{\bp_2}{\pmax - \bp_1 - \bp_2} - \frac{1}{n_2}\right] (\bp_2 - p_s) - p_s & = & 0.
	\end{eqnarray}
With the introduction of non-null $p_s$, this system is now not linear.
So, we begin with a study of a simplified setting where $n_1=n_2=2$.

\subsection{Study of the 2 ISPs, 2 CPs case}

\label{subsec:22}

First, define $x := \bp_1 / \pmax$, $y := \bp_2 / \pmax$ and $s := p_s / \pmax$. We can rewrite equilibrium conditions as
	\begin{eqnarray}
	\label{c1} -3x^2 - s - xy + ys - (s-1) x & = & 0, \\
	\label{c2} -3y^2 + s - xy - xs + (s+1) y & = & 0.
	\end{eqnarray}
Now, change variables to $u := x+y$ and $v := x-y$ to simplify these conditions to
	\begin{eqnarray}
	\label{e1} -2sv - 2u^2 - v^2 + u & = & 0, \\
	\label{e2} -3uv - 2s + v & = & 0,
	\end{eqnarray}
where \pref{e1} $=$ \pref{c1} $+$ \pref{c2} and \pref{e2} $=$ \pref{c1} $-$ \pref{c2}.
Equilibrium prices, demand and revenues are now solvable in closed form\footnote{These expressions are complicated and of no extra-computational interest.}. An important observation we can make at this point is given by the following theorem:

	\begin{theorem}
	\label{thm:side-threshold}
	When $n_1=n_2=2$, there is an interior NEP iff
	\begin{equation*}
	\left| \frac{p_s}{\pmax} \right|
	\ \leq \ \max_{x \in [\frac14, \frac12]} \sqrt{\frac{(1-x)(1-2x)^2(4x-1)}{36x}}
	\ \approx \ 4.64 \% .
	\end{equation*}
	\end{theorem}
In other words, regulated side payments can only occur \emph{to a small extent} ($|p_s| < 4.64\%$ of $\pmax$), otherwise there will be {\em no interior} NEP, which means one of the two groups of players will opt out of the usage-based pricing game.

\begin{proof}
Define $a := s/v + 1/2$. Then, equation \pref{e2} yields
	\begin{equation}
	\label{u_from_a} u = \frac23 (1-a).
	\end{equation}
Expressing $u$ and $v$ as functions of $a$ in \pref{e1}, and given that $a \neq \frac12$ ($v$ is finite), we have the following necessary condition:
	\begin{equation*}
	(6s)^2 = \frac{1}{a} (1-a) (2a-1)^2 (4a-1),
	\end{equation*}
which we can rewrite as
	\begin{equation}
	\label{abs-s} |s| \ = \ \frac{|2a - 1|}{6} \sqrt{\left(\frac{1}{a} - 1\right)(4a-1)} \ =: \ g(a).
	\end{equation}
Now, let us derive necessary conditions of the domains where our variables live.
\begin{itemize}
\item $(0 < x,y < 1)$: revenues are positive;
\item $(0 < u < 1)$: demand is positive;
\item $(-1 < v < 1)$: $x$ and $y$ are in $]0,1[$;
\item $(-\frac12 < a < 0)$ or $(\frac14 \leq a \leq 1)$: comes from \pref{u_from_a} and \pref{abs-s};
\item $(\frac14 \leq a \leq 1)$: $|v| < 1$ means that $|a - 1/2| > |s|$. Replacing $|s|$ by expression \pref{abs-s} yields $(1/a - 1)(4a - 1) < 9$, which is impossible when $-\frac12 < a < 0$.
\item $(\frac14 \leq a \leq \frac12)$: suppose $a > \frac12$ and $s>0$. From expression \pref{abs-s}, we then have
	$$s> \frac13(1-a)(2a-1) \ \Rightarrow \ u < v \ \Rightarrow \ y < 0,$$
which is impossible. Similarly, $a > \frac12$ and $s<0$ would imply $x<0$.
\end{itemize}
According to \pref{abs-s}, $|s|$ cannot be greater than the maximum value of $g$ on $\left[\frac14, \frac12\right]$, which is $\approx 4.6\%$. We showed that this condition is necessary: one can check with no harm that it is also sufficient, \ie any value of $|s| \leq \max\{g(a), \ a \in [\frac14, \frac12]\}$ yields two possible values for $a$, which in turn give two set of prices $x,y$ corresponding to positive equilibrium demand and revenues.
\end{proof}

There are two solutions to \pref{side-equ1} and \pref{side-equ2}. For any solution $\vbp^* = (\bp_1^*, \bp_2^*)$ to this system,
	\begin{equation*}
	\pdd{U_{ki}}{p_{ki}}(\vbp^*)
		\ = \ \frac{-2\,U_{ki}(\vbp^*)}{(\bp_k^* + (-1)^{k+1} p_s)(\pmax - \bp_1^* - \bp_2^*)} \ < \ 0 \textrm{ for } k=1,2,
	\end{equation*}
given that $\bp_k^* + (-1)^{k+1} p_s > 0$ for $k=1,2$ (in a valid solution, players paying side payments cannot get negative revenues). Thus, both the critical points we computed are also interior Nash Equilibrium Points (NEPs).

	\begin{figure}[ht]
	\centering
	\subfloat[]{%
		\includegraphics[width=0.48\textwidth]{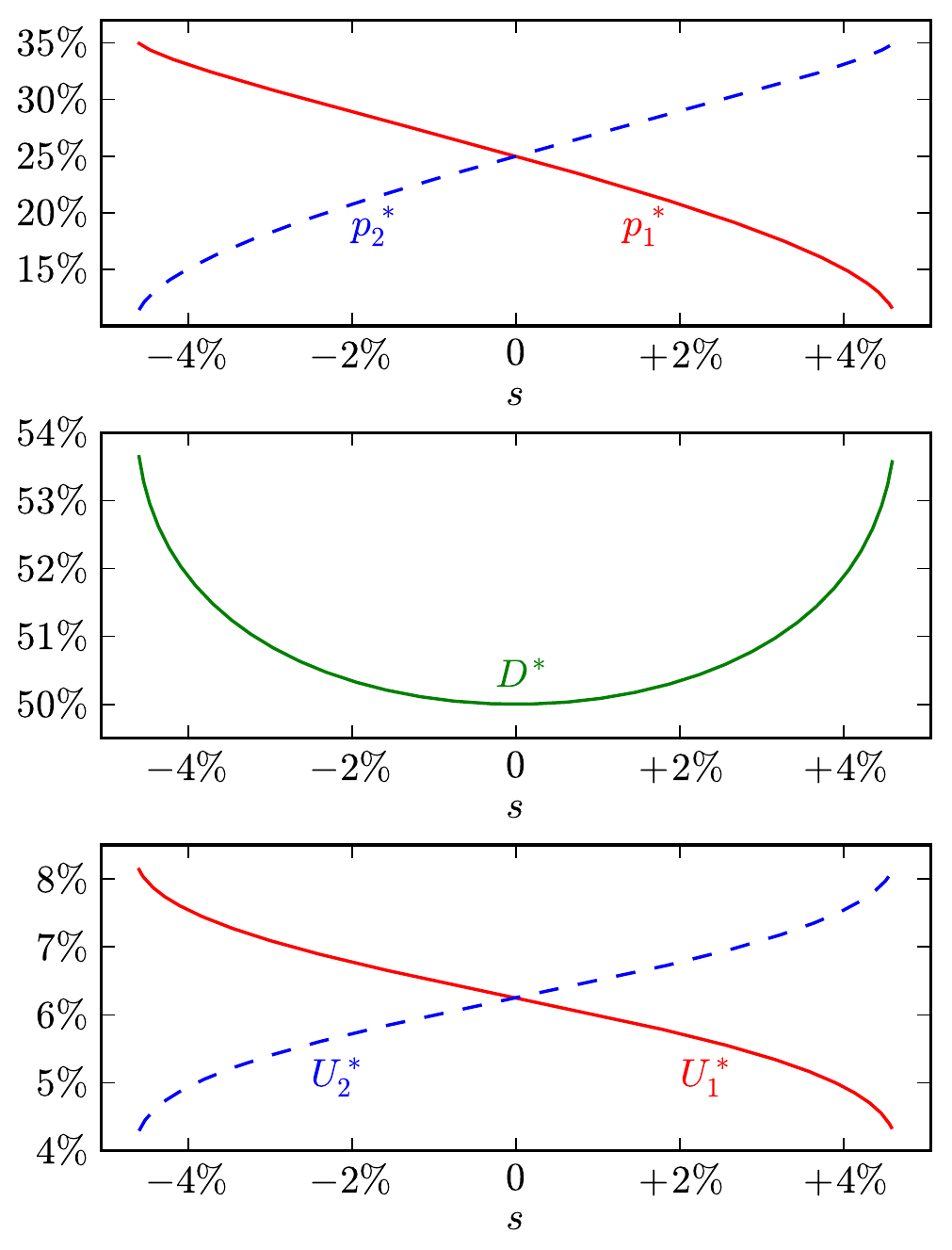}
		\label{fig:side-nep1}}
	\subfloat[]{%
		\includegraphics[width=0.48\textwidth]{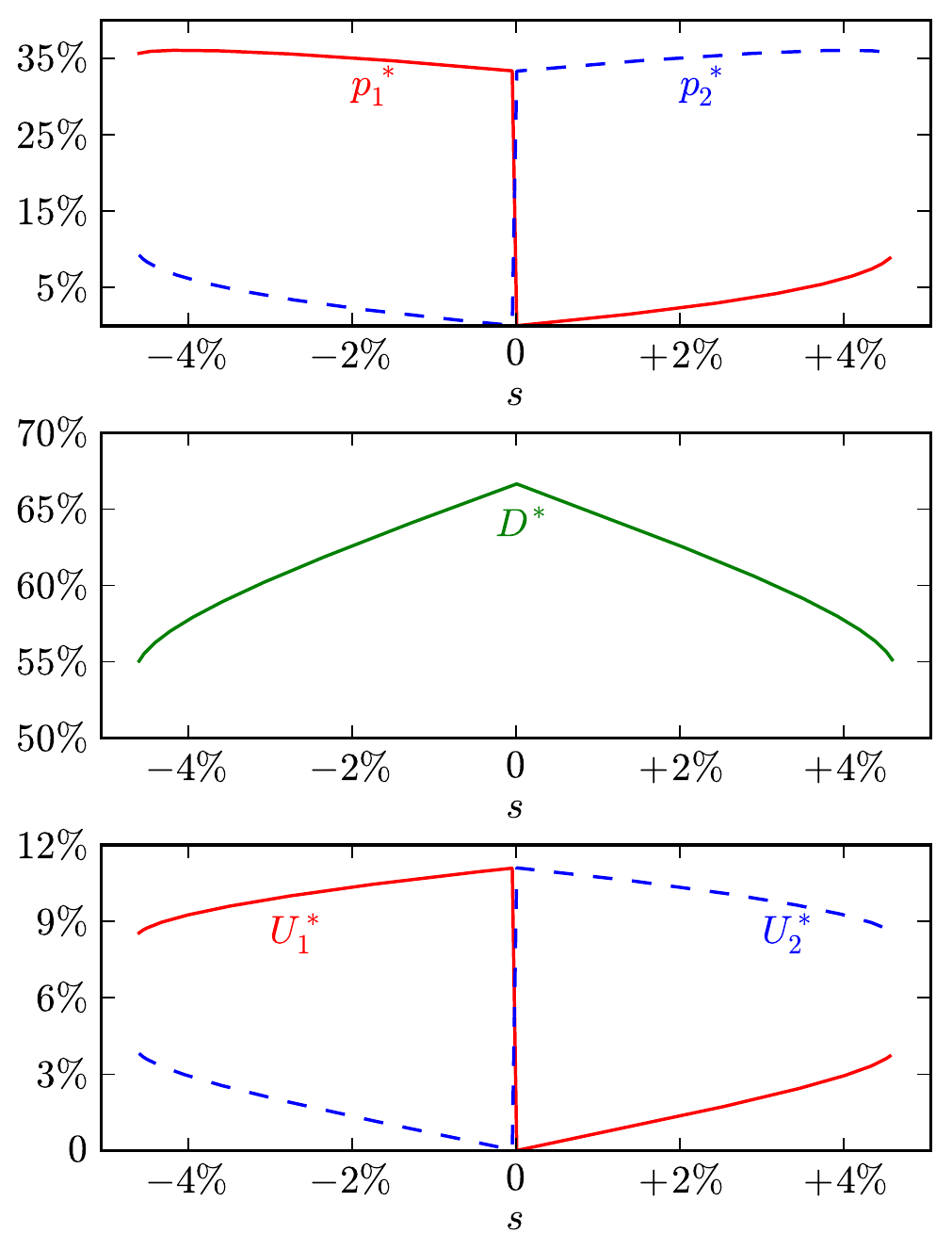}
		\label{fig:side-nep2}}
	\caption{Demand and revenues at $\NEP{1}$ and $\NEP{2}$.}
	\end{figure}

Generally, additional NEPs may exist on the boundary of the play-action space. Demand and revenues at $\NEP{1}$ and $\NEP{2}$ are shown in Figures \ref{fig:side-nep1} and \ref{fig:side-nep2}. Note that
	\begin{itemize}
	\item $\NEP{1}$ is consistent with the results of the non-discriminating setting: when $s=0$, $\bp_k^*=\pmax/4$, $D^*=\Dmax/2$ and $U_k^{i*} = \Umax/16$ for $k=1, 2$. Otherwise, $p_s$ has a rather unexpected impact on equilibrium revenues: more side payments yield \emph{decreased} revenues for those who receive them.
	\item $\NEP{2}$ does not exist when $s=0$ (there is a discontinuity in equilibrium prices at this point). Again, providers receiving side payments eventually get much less revenues than the others. Yet, unlike $\NEP{1}$, here their revenues increase with $p_s$.
	\end{itemize}
Both interior NEPs share the same ``paradox'': providers receiving side payments eventually achieve \emph{less} revenue than the others.

\subsection{Convergence to equilibrium}
\label{sec:side-conv}

Here we take $s>0$ (the roles of ISPs and CPs are swapped for $s<0$). Assume all providers act independently under a best-response behavior. Thus, the vector field $$(\bp_1, \bp_2) \mapsto \left(\pd{U_{1i}}{p_{1i}}(\bp_1, \bp_2), \pd{U_{2j}}{p_{2j}}(\bp_1, \bp_2)\right)$$ is an appropriate indicator of the aggregate ``trends'' of the system, see Figure \ref{fig:side-quiver}. 
So, if $\bp_1 > \bp_1^*(\NEP{2})$, the system is attracted by $\NEP{1}$; otherwise, unless $\bp_1$ is precisely equal to $\bp_1^*(\NEP{2})$, the system is attracted to the boundary $\NEP{B}$ (where usage-based revenues for ISPs come only from side payments), \ie $\NEP{2}$ is an unstable (saddle) point.

	\begin{figure}[ht]
	\centering
	\includegraphics[width=0.65\textwidth]{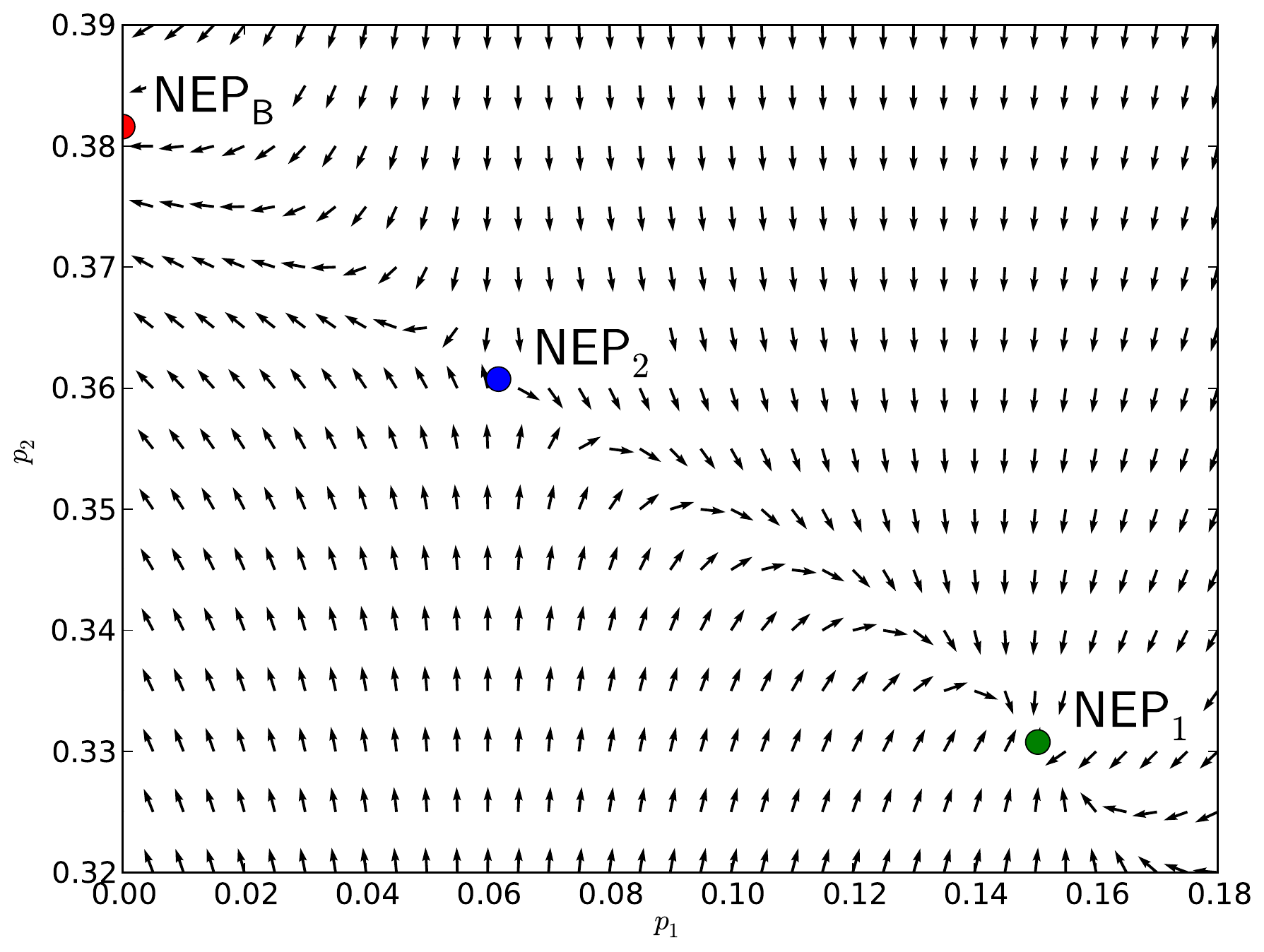}
	\caption{``Macroscopic'' trends of the system for $n_1=n_2=2$ and $s=4\%$.}
	\label{fig:side-quiver}
	\end{figure}

Solving \pref{side-equ2} with $\bp_1=0$ yields
	\begin{equation*}
	\bp_2^*(\textrm{$\NEP{B}$}) = \frac{\pmax}{6} \left(1 + s + \sqrt{s^2 + 14s + 1}\right),
	\end{equation*}
where corresponding expressions for demand and revenues follow directly. At the boundary $\NEP{B}$, demand is higher than at $\NEP{1}$ or $\NEP{2}$ while ISP revenues turn out to be lower (and CP revenues higher) than at $\NEP{2}$ (see Figure \ref{fig:side-bep}).

	\begin{figure}[t]
	\centering
	\includegraphics[width=0.6\textwidth]{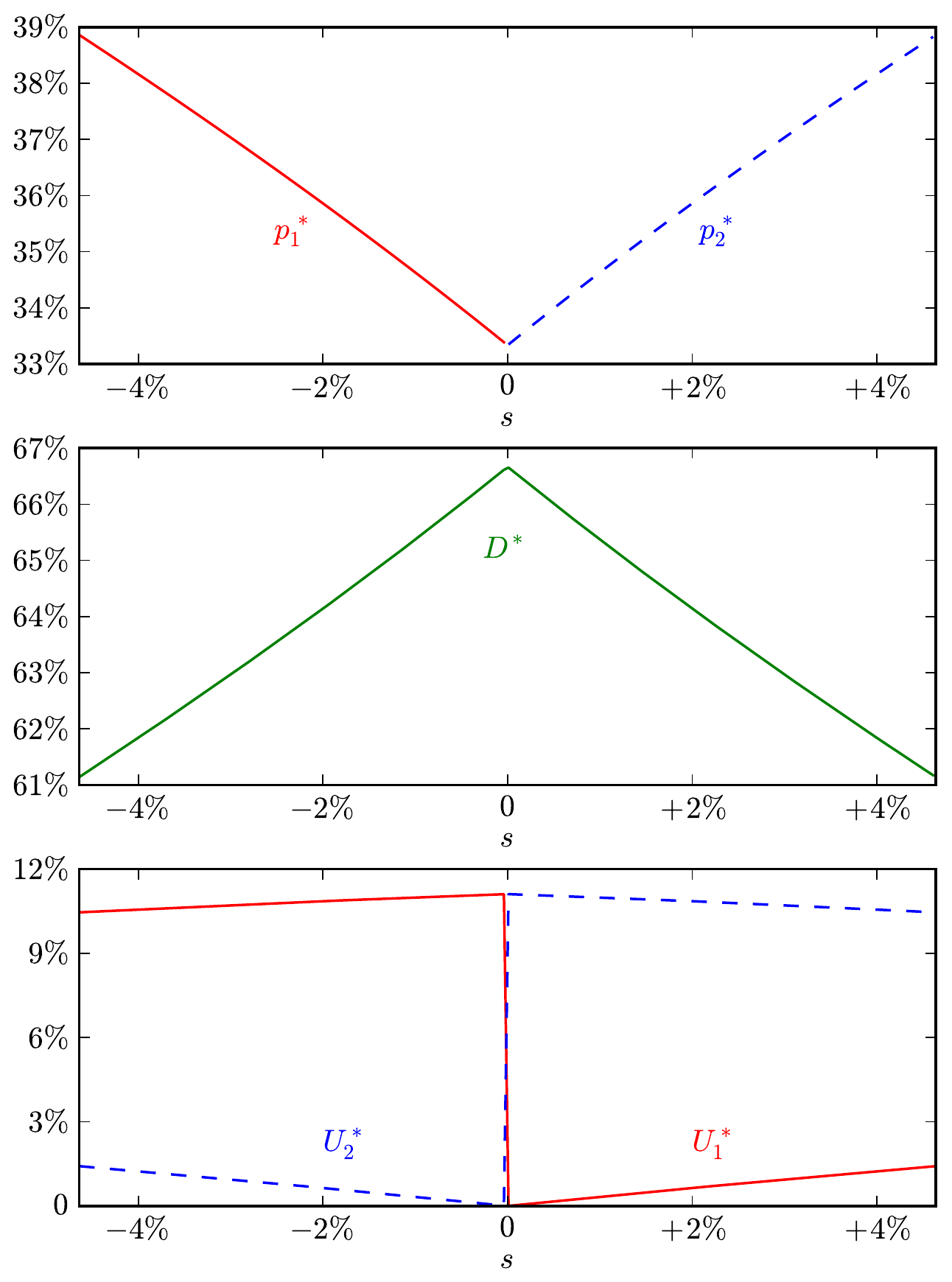}
	\caption{Demand and revenues at $\NEP{B}$.}
	\label{fig:side-bep}
	\end{figure}

\subsection{Generalization to $n_1=n$ ISPs and $n_2=n$ CPs}

When there are the same number $n$ of CPs and ISPs, necessary conditions for a NEP can be written as follows (using the change of variables of section \ref{subsec:22}):
	\begin{eqnarray*}
	2u(1 - u) - 2nsv - n (u^2 + v^2) & = & 0, \\
	nu(s-v) - 2s (n+u-1) -uv + v & = & 0.
	\end{eqnarray*}
This system is tractable for any $n \geq 2$. We computed its solutions for all $n < 10$ and observed that the conclusions of the $n=2$ study above also hold for greater values of $n$, \ie
\begin{itemize}
	\item there are two interior NEPs: one ($\NEP{2}$) is repulsive (the aggregate behavior of one group of providers is to avoid it) while the other one ($\NEP{1}$) is attractive;
	\item if side payment receivers $k$ have a mean price $\bp_k > \bp_k^*(\NEP{2})$, the system will converge to $\NEP{1}$; otherwise, it will go to a border equilibrium $\NEP{B}$ where they have a null usage-based price;
	\item equilibrium demand and revenue curves share the shape of those on Figures \ref{fig:side-nep1} and \ref{fig:side-nep2}, scaled by a factor depending on $n$.
\end{itemize}
Additional details can be found at \cite{scripts}. The number $n$ of providers impacts the maximum value of $|s|$ (see Theorem \ref{thm:side-threshold}) and scales equilibrium prices, demand and revenues. Sample figures are given in Table \ref{table:side-factors}.

	\begin{table}[ht]
	\begin{center}
	\begin{tabular}{|c|c|c|c|c|c|}
	\hline
	$n$ & $\max |s|$ & $\min \frac{D^*}{\Dmax}$ & $\max \frac{D^*}{\Dmax}$
	& $\max \frac{\bp^*}{\pmax}$ & $\max \frac{U^*}{\Umax}$ \\
	\hline
	2 & $4.7\%$ & $50\%$ & $66\%$ & $36\%$ & $11\%$ \\
	3 & $2.2\%$ & $60\%$ & $75\%$ & $27\%$ & $6.2\%$ \\
	4 & $1.3\%$ & $67\%$ & $80 \%$ & $22\%$ & $4.0\%$ \\
	5 & $0.83\%$ & $71\%$ & $83\%$ & $19\%$ & $2.8\%$ \\
	6 & $0.59\%$ & $75\%$ & $86\%$ & $16\%$ & $2.0\%$ \\
	7 & $0.44\%$ & $77\%$ & $87\%$ & $14\%$ & $1.6\%$ \\
	8 & $0.34\%$ & $80\%$ & $88\%$ & $13\%$ & $1.3\%$ \\
	9 & $0.27\%$ & $81\%$ & $90\%$ & $11\%$ & $1.0\%$ \\
	\hline
	\end{tabular}
	\caption{Impact of $n$ on equilibria properties.}
	\label{table:side-factors}
	\end{center}
	\end{table}

\subsection{Summary}

For reasonable values of the number of competing providers, we saw that the introduction of side payments in the model yielded a game with two possible outcomes:
\begin{itemize}
	\item If initial prices of side payment receivers are high enough, providers will reach an equilibrium where receivers get less revenue than payers (the higher the side payments, the lower the revenue). Yet, this is the best compromise for them.
	\item Otherwise, receivers will be constrained into setting their usage-based fees to zero, depending only on side payments for their usage-based revenues. This is the worst solution for them.
\end{itemize}
In both cases, the paradox of side payments is that they act as a \emph{handicap} for those who receive them.

\section{Application neutrality}
\label{sec:app}

Now, let us consider to what extent ISPs should be allowed to perform price discrimination depending on the application in use (\eg video chat, media streaming)? In this section, we study the impact of such discrimination in a configuration with two crude example types of applications: web surfing and P2P file sharing.

\subsection{Additional problem set-up}

We extend our model to a setting with three types of providers:
\begin{enumerate}
	\item ISPs, providing last-mile access to the Internauts,
	\item Web Content Providers (Web CPs), \eg ~ search engine portals (recall all providers of any given type are deemed identical, so we assume all Web CPs provide the same type of client-server HTTP content as well), and
	\item P2P Content Providers (P2P CPs), \eg private P2P networks operated in cooperation with copyright holders.
\end{enumerate}
Users choose an ISP, a Web CP and a P2P CP. To access web (resp. P2P) content, they pay usage-based fees to both their ISP and their Web CP (resp. P2P CP). These groups are not coalitions: in a group, each provider acts independently to maximize their own revenue.

In a neutral setting, the $\th{i}$ ISP charges a single price $p_{1i}$ for all types of traffic, while otherwise it may set up two different prices $p_{12,i}$ and $p_{13,i}$ for HTTP and P2P traffic respectively. Denote by $p_{2j}$ (resp. $p_{3j}$) the usage-based price charged by  the $\th{j}$ Web CP (resp. P2P CP). We introduce two separate demand-response profiles for the two types of content: when ISP $i$, Web CP $j$ and P2P CP $l$ are chosen, demands for HTTP and P2P content are, respectively,
	\begin{eqnarray*}
	D_2 & = & \Dkmax{2} - d_2 (p_{12,i} + p_{2j}), \\
	D_3 & = & \Dkmax{3} - d_3 (p_{13,i} + p_{3l}),
	\end{eqnarray*}
with $p_{12,i} = p_{13,i} = p_{1i}$ in the neutral setting. As previously, define $\pkmax{k} := \Dkmax{k} / d_k$.

The portion of users committed to the $\th{i}$ provider of the $\th{k}$ group is still modeled as (\ref{eqn:sigma});
we will see in \ref{app-non-neutral} how to generalize this to ISPs charging two different prices instead of one. Revenues for ISP $i$, Web CP $j$ and P2P CP $l$ are given by
	\begin{eqnarray*}
	U_{1i} & = & \sigma_{1i} \, (D_2 \, p_{12,i} + D_3 \, p_{13,i}), \\
	U_{2j} & = & \sigma_{2j} \, D_2 \, p_{2j}, \\
	U_{3l} & = & \sigma_{3l} \, D_3 \, p_{3l}.
	\end{eqnarray*}
Finally, we define the normalized sensitivity to usage-based pricing $\alpha$ and the maximum prices ratio $\gamma$,
	\begin{equation}
	\label{eqn:alpha-gamma}
	\alpha := \frac{d_2}{d_2 + d_3} \mbox{ and }  \gamma := \frac{\pkmax{2}}{\pkmax{3}},
	\end{equation}
and make the following assumptions:
	\begin{itemize}
	\item $\alpha \geq 1/2 \Leftrightarrow d_2 > d_3$: consumers are more sensitive to usage-based pricing for
	web content than for file sharing.
	\item $\gamma < 1 \Leftrightarrow \pkmax{2} < \pkmax{3}$: customers are ready to pay more for content
	exchanged on P2P sharing systems (movies, music, {\em etc.}) than for web pages.
	\end{itemize}

{

\subsection{Monopolistic players}
\label{subsec:app-monopolistic}

Now assume there is only one ISP, one Web CP and one P2P CP. This is the case when, in any group, either there is no competition or all providers have decided to form a coalition. Closed solutions are easy to derive here.

In the non-neutral setting, the ISP plays two independent ``ISP vs. CP'' games and the equilibrium is thus given by
	\begin{equation*}
	p_{1k}^* = p_k^* = \frac{\pkmax{k}}{3} \ ; \ 
	D_k^* = \frac{\Dkmax{k}}{3} \ ; \ 
	U_k^* = \frac{\Ukmax{k}}{9}
	\end{equation*}
for $k=2,3$, with $U_1^* = U_2^* + U_3^*$. See \cite{arXiv} for further discussion of the ``ISP vs. CP'' game.

In the neutral setting, the ISP has to find a compromise between the two applications, which is at the NEP
	\begin{equation}
	\label{app-monop-p1}
	p_1^* = \alpha\,\frac{\pkmax{2}}{3} + (1-\alpha)\,\frac{\pkmax{3}}{3}.
	\end{equation}
Let us define $\Delta_{\sf max} := \pkmax{3} - \pkmax{2}$ and $\mu := \frac{2}{1/d_2 + 1/d_3}$ the harmonic mean of $d_2$ and $d_3$. Equilibrium prices set by the two other providers are
	\begin{equation*}
	p_2^* = \frac{\pkmax{2}}{3} - (1-\alpha)\,\frac{\Delta_{\sf max}}{6} \ ; \ 
	p_3^* = \frac{\pkmax{3}}{3} + \alpha\,\frac{\Delta_{\sf max}}{6}.
	\end{equation*}
Neutrality yields lower demand for web content and higher demand for file sharing:
	\begin{equation*}
	D_k^* = \frac{\Dkmax{k}}{3} - (-1)^k \frac{\mu\,\Delta_{\sf max}}{12} \textrm{ for } k=2,3.
	\end{equation*}
Equilibrium revenue for the ISP is
	\begin{equation*}
	U_1^* \ = \ \frac{(\Dkmax{2} + \Dkmax{3})^2}{9 (d_2 + d_3)} \ < \ \frac{\Ukmax{2}}9 + \frac{\Ukmax{3}}{9},
	\end{equation*}
less than in the non-neutral setting. Other players' revenues are given by
	\begin{equation*}
	U_k^* = \frac{\Ukmax{k}}{18} \left(1 - (-1)^k\frac{\mu \Delta_{\sf max}}{\Dkmax{2}}\right)^2 \textrm{ for } k=2,3.
	\end{equation*}

These expressions show that, while the P2P CP is better off in a neutral setting, both the ISP and the Web CP prefer the non-neutral configuration{ (see Figure \ref{fig:app-monop})}.
	\begin{figure}[ht]
	\centering
	\includegraphics[width=0.6\textwidth]{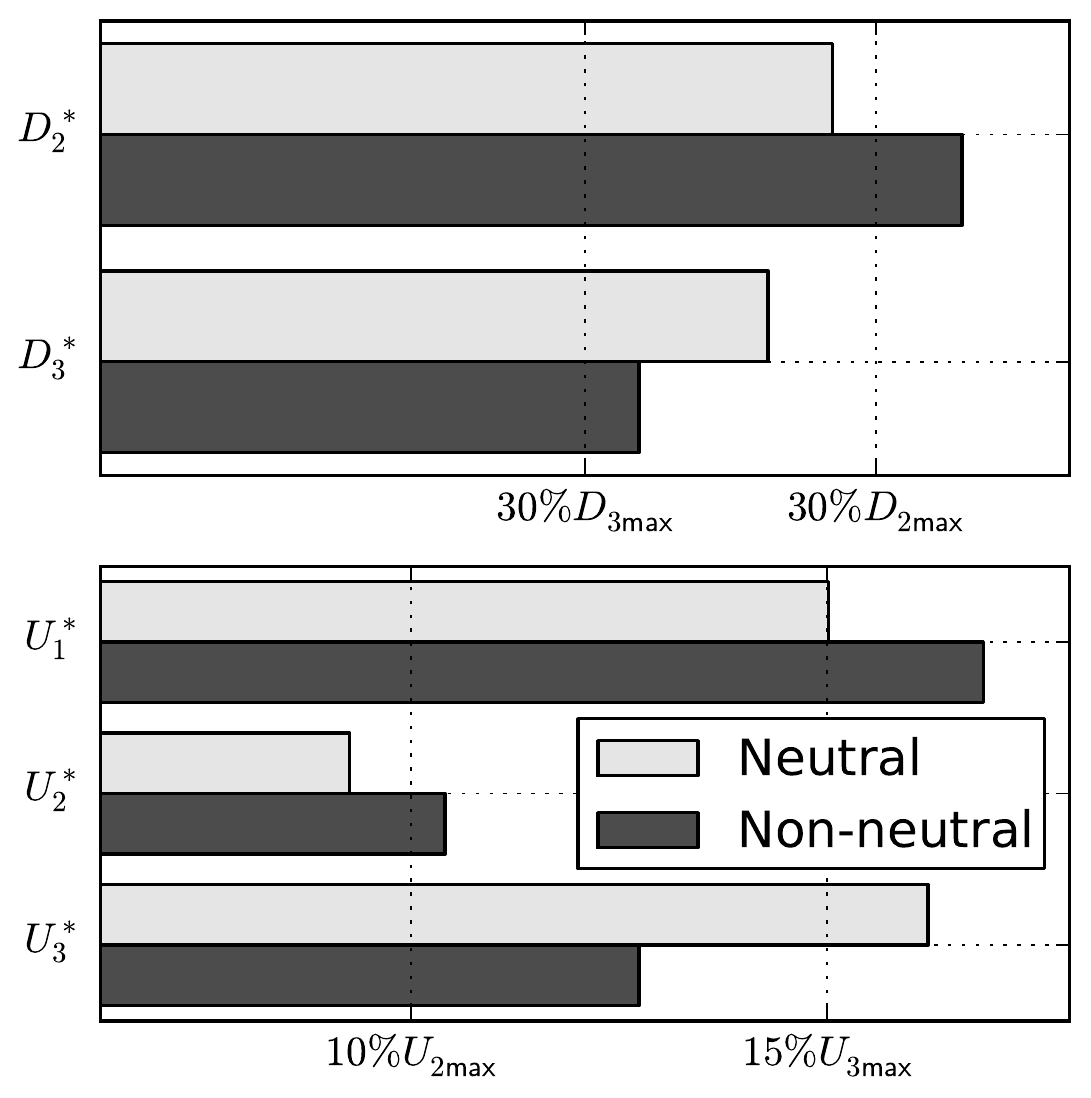}
	\caption{Impact of regulation with monopolistic providers ($\alpha=.8$ and $\gamma=.4$).}
	\label{fig:app-monop}
	\end{figure}
Also, an interesting fact to point out is that $\pkmax{3}$ cannot be too high:
	\begin{theorem}
	\label{thm:app-monop-neutral}
	In the neutral setting with monopolistic providers, there is an interior NEP iff
		\begin{equation}
		\label{app-monop-bound}
		\pkmax{3} \ < \ \left(1 + \frac{2}{1-\alpha}\right) \, \pkmax{2}.
		\end{equation}
	\end{theorem}

\begin{proof}
\pref{app-monop-bound} is equivalent to $p_2^* > 0$. If it does not hold, the Web CP will have no choice but to set his usage-based price to zero, thus opting out of the game.
\end{proof}

This condition is a consequence of the ``compromise'' sought by the ISP: if $\pkmax{3} \gg \pkmax{2}$, equation \pref{app-monop-p1} tells us $p_1^*$ will be greater than the maximum price consumers are ready to pay for web content.

}

\subsection{Neutral setting}
\label{app-neutral}
Here consider the setting of non-monopolistic provi\-ders (\ie $n_k > 1$ for $k \in \{1,2, 3 \}$) where application neutrality is enforced. In particular, $U_{1i} = \sigma_{1i} (D_2 + D_3) p_{1i}$.
Utilities' derivatives are then:
	\begin{eqnarray*}
	\pd{U_{1i}}{p_{1i}}(\bp_1, \bp_2, \bp_3)
		& = & \left[ \frac{1}{n_1 \bp_1} - \frac{1}{\alpha \wtD_2 + (1-\alpha) \wtD_3}\right] U_{1i}, \\
	\pd{U_{ki}}{p_{ki}}(\bp_1, \bp_2, \bp_3)
		& = & \left[\frac{1}{n_k \bp_k} - \frac{1}{\wtD_k}\right] U_{ki} \textrm{ for } k=2,3,
	\end{eqnarray*}
where $\wtD_k := D_k / d_k = \pkmax{k} - \bp_1 - \bp_k$ for $k=2,3$. Therefore, a NEP must be solution to the linear system:
\small
	\begin{eqnarray*}
	(n_1 + 1)\,\bp_1 + \alpha\,\bp_2 + (1-\alpha)\,\bp_3 
& = & \alpha\,\pkmax{2} + (1-\alpha)\,\pkmax{3}, \\
	\bp_1 + (n_2 + 1)\,\bp_2 & = & \pkmax{2}, \\
	\bp_1 + (n_3 + 1)\,\bp_3 & = & \pkmax{3},
	\end{eqnarray*}
\normalsize
whose resolution is straightforward.
For any solution $\vbp^* = (\bp_1^*, \bp_2^*, \bp_3^*)$ of this system,
	\begin{eqnarray*}
	\pdd{U_{1i}}{p_{1i}}(\vbp^*)
		& = & \frac{-2\,U_{1i}(\vbp^*)}{\bp_1 (\alpha \wtD_2 + (1-\alpha) \wtD_3)} \ < \ 0 \\
	\pdd{U_{ki}}{p_{ki}}(\vbp^*)
		& = & \frac{-2\,U_{ki}(\vbp^*)}{\bp_k \wtD_k} \ < \ 0 \textrm{ for } k=2,3,
	\end{eqnarray*}
ensuring that $\vbp^*$ is indeed a NEP.

\subsection{Non-neutral setting}
\label{app-non-neutral}

When application non-neutral pricing is allowed, the $\th{i}$ ISPs utility is $U_{1i} = \sigma_{1i} (D_2 p_{12,i} + D_3 p_{13,i})$, where $\sigma_{1i}$ refers to the portion of users gathered by ISP $i$ given his prices $p_{12,i}$ and $p_{13,i}$. There are different ways to generalize equation \pref{eqn:sigma} to multiple criteria: \eg one could apply $\sigma$ to the mean price $(p_{12,i} + p_{13,i})/2$ or model $\sigma_{1i}$ as a convex combination of $\sigma_{12,i}$ and $\sigma_{13,i}$. We chose
	\begin{equation}
	\label{app-non-neutral-sigma}
	\sigma_{1i} \ := \ \sigma(i, \vec{\wtp}_1) \ = \ \frac{1/\wtp_{1i}}{\sum_{j=1}^{n_1} 1/\wtp_{1j}}
	\end{equation}
where $\wtp_{1i} := \sqrt{\alpha\gamma}\,p_{12,i} + (1 - \sqrt{\alpha\gamma})\,p_{13,i}$. That is, we apply the original stickiness model \pref{eqn:sigma} to a combined price $\wtp_{1i}$ defined as a convex combination of $p_{12,i}$ and $p_{13,i}$.
This choice, particularly the geometric mean $\sqrt{\alpha\gamma}$ in \pref{app-non-neutral-sigma}, is motivated by the following considerations:
 $\sigma_{1i}$ still satisfies the properties expected for a stickiness function (see section \ref{sec:setup});
the weight of $p_{12,i}$ in the combination is increasing in $\pkmax{2}$ and $d_2$, and similarly the weight of $p_{13,i}$ is increasing in $\pkmax{3}$ and $d_3$; and
the resulting model is solvable in closed form.
	
In this model, utilities' derivatives for CPs are the same as in \ref{app-neutral} (regulations only affect the ISPs) while
	\begin{eqnarray*}
	\pd{U_{1i}}{p_{12,i}} & = & \left[\frac{\alpha\,(\wtD_2 - \bp_{12})}{\alpha \wtD_2 \bp_{12} + (1-\alpha) \wtD_3 \bp_{13}} - \left(1-\frac{1}{n_1}\right) \frac{\sqrt{\alpha\gamma}}{\wtp_{1i}}\right] U_{1i}, \\
	\pd{U_{1i}}{p_{13,i}} & = & \left[\frac{(1-\alpha)\,(\wtD_3 - \bp_{13})}{\alpha \wtD_2 \bp_{12} + (1-\alpha) \wtD_3 \bp_{13}} - \left(1-\frac{1}{n_1}\right) \frac{1-\sqrt{\alpha\gamma}}{\wtp_{1i}}\right] U_{1i}.
	\end{eqnarray*}
Thus, any Nash equilibrium satisfies:
	\begin{eqnarray*}
	\frac{\alpha}{\sqrt{\alpha\gamma}}\,(\wtD_2 - \bp_{12})
		& = & \frac{n_1-1}{n_1 \wtp_1} (\alpha \wtD_2 \bp_{12} + (1-\alpha) \wtD_3 \bp_{13}), \\
& = & \frac{1-\alpha}{1-\sqrt{\alpha\gamma}}\,(\wtD_3 - \bp_{13}), \\
	\bp_{12} + (n_2 + 1)\,\bp_2 & = & \pkmax{2}, \\
	\bp_{13} + (n_3 + 1)\,\bp_3 & = & \pkmax{3},
	\end{eqnarray*}
where $\wtp_1 := \sqrt{\alpha\gamma}\,\bp_{12} + (1 - \sqrt{\alpha\gamma})\,\bp_{13}$.
This system can be rewritten as two polynomial equations in $\bp_{12}$ and $\bp_{13}$ which are solvable in closed form. Computations yield a single admissible solution $\vbp^* = (\bp_{12}^*, \bp_{13}^*, \bp_2^*, \bp_3^*)$ here, for which
	\begin{eqnarray*}
	\pdd{U_{1i}}{p_{12,i}}(\vbp^*)
		& = & \frac{-2\alpha\,U_{1i}(\vbp^*)}{\alpha \wtD_2 \bp_{12} + (1-\alpha) \wtD_3 \bp_{13}} \ < \ 0,
	\end{eqnarray*}
and similarly for $\pdd{U_{1i}}{p_{13,i}}(\vbp^*)$ with $2(1-\alpha)$ instead of $2\alpha$. As in \ref{app-neutral}, $\pdd{U_{2j}}{p_{2j}}(\vbp^*)$ and $\pdd{U_{3l}}{p_{3l}}(\vbp^*)$ are negative as well, ensuring that $\vbp^*$ is indeed a Nash equilibrium of the game.

\subsection{Discussion of experimental results}

In{ most of} our numerical experiments, we compared revenues at this NEP with those of the neutral scenario for $\alpha = 0.8$ and $\gamma = 0.3$. We used Sage {\cite{sage}} for our computations, and all our scripts are available  at \cite{scripts}.

First, the trend we observed in subsection \ref{subsec:app-monopolistic} (ISPs and Web CPs prefer the non-neutral setting while P2P CPs benefit from neutrality regulations) also holds when the model encompasses competition{ (see Figure \ref{fig:app-competitive-revenues})}.
	\begin{figure}[ht]
	\centering
	\includegraphics[width=0.75\textwidth]{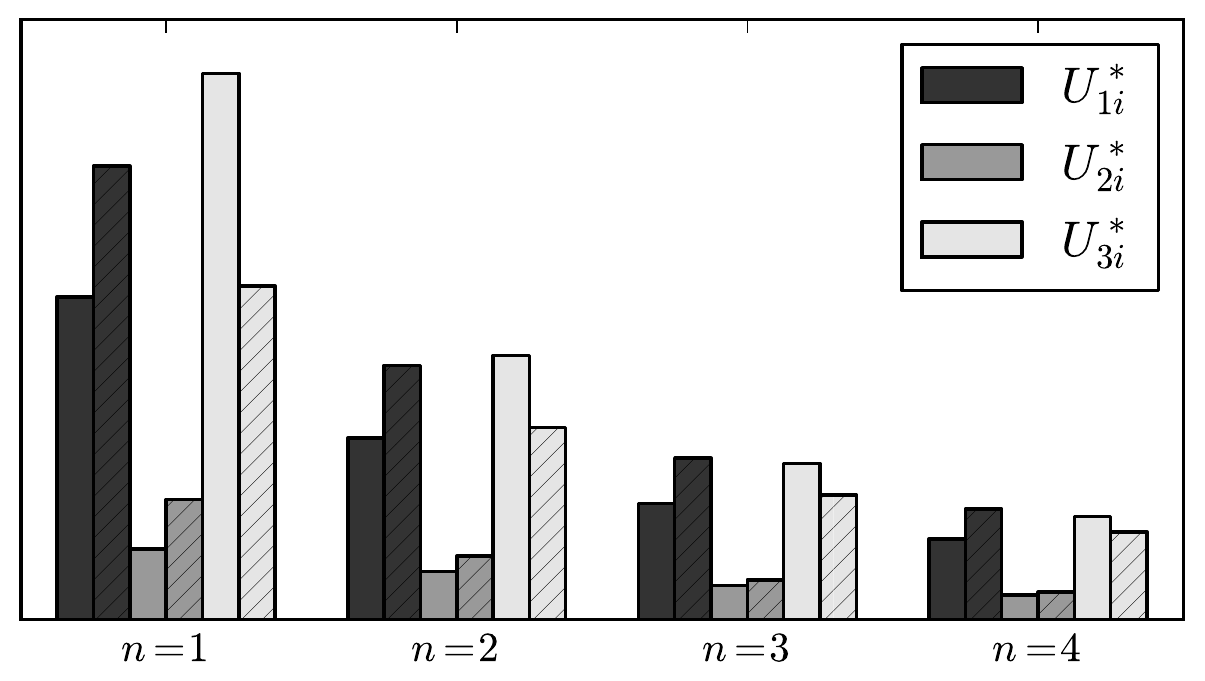}
	\caption{Revenues for providers in different settings
	($n$ is the number of providers of each type, $\alpha=0.8$ and
	$\gamma=0.3$). Plain (resp. hatched) columns correspond to
	neutral (resp. non-neutral) settings.}
	\label{fig:app-competitive-revenues}
	\end{figure}

The impact of non-neutral pricing on providers' revenues varies with competition: increased competition brings less benefit for Web CPs and less loss for P2P CPs. Yet, competition has almost no effect on the gains of ISPs (see Figure \ref{fig:app-competitive-relvar}).
	\begin{figure}[ht]
	\centering
	\includegraphics[width=0.75\textwidth]{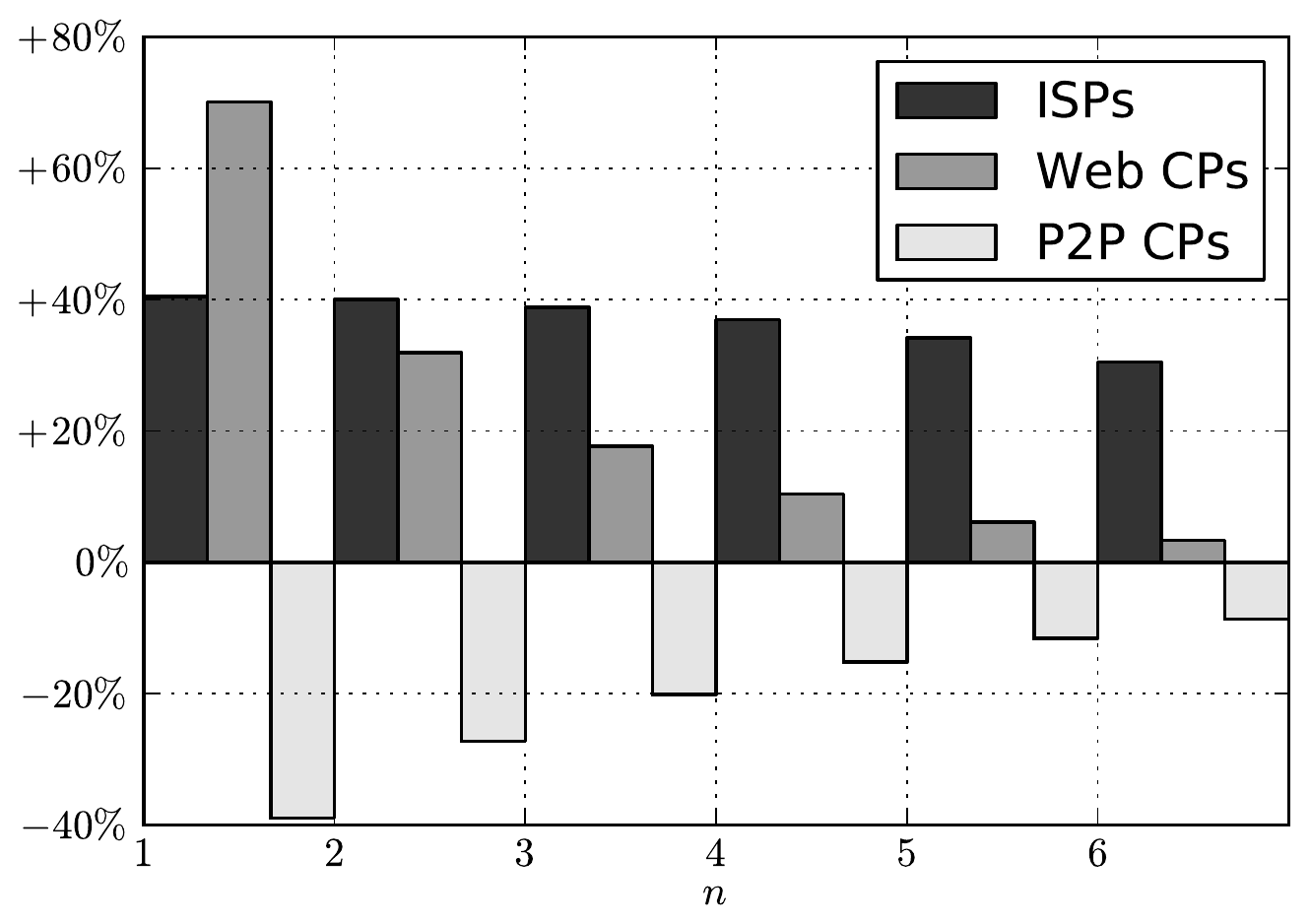}
	\caption{Relative variation in revenue, \ie the ratio of revenues (non-neutral - neutral)/neutral at the NEP, where $n$ is the number of providers of each type, $\alpha=0.8$ and $\gamma=0.3$.}
	\label{fig:app-competitive-relvar}
	\end{figure}

{We also observed a maximum price gap, as we did in subsection \ref{subsec:app-monopolistic} with Theorem \ref{thm:app-monop-neutral}, which decreases when competition increases.}

\section{Conclusions}
\label{sec:conclusion}

We presented an idealized framework to study the impact of two net-neutrality related issues, side payments and application neutrality, on the interactions among end users, ISPs and CPs. Our revenue model relied on a simple, common linear demand response to usage-based prices, and it accounted for customer loyalty.

We studied the effect of regulated side payments between the ISPs and CPs. The two possible outcomes of the competition  both showed the same paradox: side payments are actually a handicap for those who receive them insofar as they reduce Nash equilibrium revenues.

We also studied the issue of application neutrality in a simple setting involving two types of content, web content and file sharing, the latter showing lower price sensitivity and higher willingness to pay under our assumptions of relative demand sensitivity to price. Our analysis suggested that ISPs and Web CPs benefit from application non-neutral practices, while providers that enable content dissemination by P2P means are better off in a neutral setting.

{

\subsection{Future Work}
\label{subsec:fw}

Our current framework does not take into account additional advertising revenues that CPs may receive, \eg to lower their usage-based prices, or even deliver their content for free. There are multiple ways to remediate this. In \cite{arXiv}, we chose to add a fixed parameter $p_a$, so that $U_{\sf CP} = (p_{\sf CP} + p_a - p_s) D$. The underlying hypothesis is that $p_a$ is the outcome of a separate game played between advertisers and content providers. However, advertisers' willingness to pay may also depend on consumers' demand $D$, and is likely to increase significantly when targeted advertising becomes possible.

One way to remediate the shortcomings of fixed $p_a$ is to add advertisers to the game, reacting to prices $\{p_{a,i}\}$ set by CPs in a way similar to consumers, \ie with linear demand-response and advertiser stickiness. Their demand $D_a$ should be increasing with consumers' demand, and decreasing with $p_a$, \eg $D_a = D - d_a p_a$ (CPs' revenues being $U_{\sf CP} = p_a D_a + (p_{\sf CP} - p_s) D$). Also, when there are sundry types of applications, their respective levels of targeted advertising could be taken into account with $d_{a,i}$ depending on application~$i$. Though harder to solve, such a system would encompass another major component of the Internet economy.

}

\end{document}